\documentclass[onecolumn,draftclsnofoot]{IEEEtran}
\usepackage{xcolor}
\definecolor{grn}{RGB}{0,150,0}
\usepackage{epsfig,rotating,setspace,latexsym,amsmath,epsf,amssymb,bm}
\usepackage{cite,graphicx,color}
\usepackage{mathrsfs}
\usepackage{subfig}
\usepackage[export]{adjustbox}
\usepackage{multirow}

\begin{document}

%

\title{Cache Aided Communications with Multiple Antennas at Finite SNR}
%
%
%

\newcommand{\E}{\mathbb{E}}

\author{Itsik~Bergel and~Soheil~Mohajer
}

\maketitle

\begin{abstract}
We study the problem of cache-aided communication for cellular networks with multi-user and multiple  antennas at finite signal-to-noise ratio. Users are assumed to have non-symmetric links, modeled by wideband fading channels. We show that the problem can be formulated as a linear program, whose solution provides a joint cache allocation along with pre-fetching and fetching schemes that minimize the duration of the communication in the delivery phase.  The suggested scheme uses zero-forcing and cached interference  subtraction and hence allow each user to be served at the rate of its own channel. Thus, this scheme is better than the previously published schemes that are compromised by the poorest user in the communication group. We also consider a special case of the parameters for which we can  derive a closed form solution and formulate the optimal power, rate  and cache optimization. This special case shows that the gain of MIMO coded caching goes  beyond the throughput. In particular, it is shown that in this case, the cache is used to balance the users such that fairness and throughput are no longer contradicting. More specifically, in this case, strict fairness is achieved jointly with maximizing the network throughput.


\end{abstract}


%
\IEEEpeerreviewmaketitle

\section{Introduction}

Network traffic has rapidly increased, over both wired and wireless networks, in recent years. This overwhelming growth is mostly due to the demands for broadband data. In particular, video delivery accounts for a major growth of traffic on both mobile  \cite{cisco2017global} and wireline networks \cite{sandvine14}. Two unique characteristics of 
of video contents are \textsf{(i)} popular files are repeatedly requested by multiple users; and \textsf{(ii)} unlike general web usage, video request has a prime time.
These unique properties provide an opportunity for storing the data at local caches during the off-peak hours of the network, and serve a request at the peak hours\cite{huston1999web}. 


In a pioneering work \cite{maddah2014fundamental}, Maddah-Ali and Niesen showed that caching gain is not limited to the local cache size at individual users. More importantly, caching a packet at User $1$, even if it is only requested by another User $2$, provides an opportunity for \emph{multicasting combined packets}, which can simultaneously serve both Users $1$ and $2$. It is shown that this scheme offers a global gain which scales with the aggregate size of the caches distributed across all  users in the network. 

This scheme was further generalized to multiple transmit antennas \cite{shariatpanahi2017multi}, where they showed that a network can achieve $N+M$ degrees of freedom (DoF), where $N$ is the number of antennas and $M$ is the number of copies of  complete dataset stored across over all users. This is a  significant gain, in contrast to only $N$ DoF, achievable with $N$ antennas and no caching. 
This potential gain is substantial, even in spite of the current trend of massive MIMO \cite{larsson2014massive, lu2014overview}, which calls for the use of antenna arrays with many elements: especially due to the cost of antennas arrays to be deployed. 
In contrast, use of a small cache at each mobile comes with very low cost, and these memories easily scale up to a total size that can offer significant gains.


Nevertheless, the existing works in this area only focus on isolated scenarios, or limited to DoF characterization (i.e., asymptotically high SNR). Thus, there is a big gap to cover before we can understand and fully exploit the role of cache-enabled communication in cellular networks. 
In practice, users are located at different distances and are subject to power attenuation, and fading. The optimal use of caching in such a multi-antenna scenario is still unknown. 


In this work, we demonstrate the achievability of the caching gain in the finite SNR regime and present closed form expressions for the performance in a special case. These results demonstrate a fascinating phenomenon, in which the cache contributes both for throughput and  fairness. Recalling that in general fairness in wireless networks comes at the price of reduced throughput, the proposed scheme brings a situation in which strict fairness is achieved through maximizing the total network throughput. In other words, the caching allows a natural balancing of the load between the users. Thus, caching brings two distinct advantages: it increases the network throughput  and  it balances between the data rate of the different users to improve fairness. 

A short illustrative example that demonstrate the core of the proposed scheme is given in Section~\ref{sub: Example}.

\subsection{Related works}
\emph{Coded caching} \cite{maddah2014fundamental} is a novel data delivery technique to exploit the aggregate cache in the network rather than individual memory available at each user. 
In general, we have a network with $U$ users and a set of $F$ files at the server, all of the same (normalized) length. Each user is equipped with a storage memory to store a fraction of the packets of each file during the \emph{placement phase}. Cache placement occurs prior to users' requests, and hence it is performed regardless and independent of the requests.  At the beginning of the delivery phase, each user requests for one file from the dataset,  and the server broadcasts a message (a sequence of packets) to simultaneously serve all user requests. That is, each user $u$ should be able to recover its desired file  from the received message and its cached information. The ultimate goal is to minimize the duration of time required to serve all users.


A \emph{global caching gain} can be realized when a single transmission can serve multiple users: A packet requested by user $u$ and not cached at his local memory can be \emph{combined} with any other packet cached at the user, since it allows the user  to null the \emph{interference} using its local memory. Such a combined packet is simultaneously useful for some other user $u'$, if it contains a packet requested by $u'$, and all other interfering components are cached in $u'$.

Several interesting bounds have been proposed to fully characterize the rate-memory tradeoff of coded caching \cite{ghasemi2017improved, sengupta2015improved, wang2016new, ajaykrishnan2015critical, tian2015note}. Under ideal assumptions and uncoded prefetching, it is shown that the proposed scheme of \cite{maddah2014fundamental} is information-theoretically optimal for some range of parameters \cite{wan2016optimality}. Optimality of (a slightly modified version of) this scheme is proved in \cite{yu2016exact} for arbitrary parameters and for both average and worst-case demand scenarios.

In general, the caching gain can be improved by allowing coded pre-fetching (referring to jointly coding across files to be placed at users’ cache), at the price of complexity of the system \cite{chen2014fundamental, tian2016caching, sahraei2016k, amiri2017fundamental}. In spite of developing
better achievability schemes, and several efforts in tightening the outer bound of the rate-memory tradeoff for caching with general placement \cite{wang2016new, wang2017improved, yu2017characterizing}, the problem is still not fully solved. 
One of the  main advantages of uncoded pre-fetching is a rather simple handling of practically-relevant asynchronous demands, without increasing the communication rates \cite{maddah2015decentralized}, and hence we only focus on uncoded placement in this work.


Recently, the attention of the community has been shifted towards the practical aspects of coded caching, and their adoption in wireless networks. In particular, \cite{bidokhti2017gaussian, gregori2015joint, yang2016content, bidokhti2016noisy,  bidokhti2016erasure, ghorbel2016content} study coded caching in wireless networks in the presence of fading and/or erasure channels. Coded caching in wireless networks with multiple antennas at transmitters and/or receivers is considered in \cite{yang2016content, shariatpanahi2016multi, shariatpanahi2017multi, ngo2018scalable}. In particular, a homogeneous (with statistically identical channel gains) MISO network is considered in  \cite{ngo2018scalable}, where a mixed communication scheme is proposed to combine spatial multiplexing and multicasting, and improve the gain as the number users grows.

Employment of coded caching in wireless networks, and in particular in cellular networks, requires addressing several practical issues. In a realistic system, each user has a channel with different statistics and capacity. Cache allocation should be optimized  depending on network traffic, user's channel quality, user's available storage, and other network characteristics. 
Coded caching for heterogeneous networks with different channels and rates for users (in the delivery phase) is studied in \cite{ibrahim2017optimization} for networks with single transmit and receive antennas. In \cite{ibrahim2017optimization}, each packet transmission is subject to the rate of the weakest user, among those supposed to decode the packet. 

One of the fundamental distinctions of our work is the exploitation of spatial diversity for the delivery phase. In particular, joint coding of the packets can be performed over-the-air, instead of at the transmitter: The transmitter sends different packets along various spatial directions. Each end-user will receive a combination of the transmit packets. The interfering packets are either nulled over the air by zero-forcing, or suppressed at the receiver using the cache content. 

Hence, the rate of each packet is only limited by the channel capacity of the intended user. A rather similar phenomena is observed in the single antenna case, by using multiple nested codebooks \cite{bidokhti2017benefits, bidokhti2018state, bidokhti2017gaussian}. 
In MIMO setting, however, this can be done naturally, since each message is sent along a different spatial direction, and users can suppress the effect of the undesired but cached messages from the received signal, even before the decoding process starts. 
This is an important characteristics of MIMO caching systems, which is further elaborated below.

The main contribution of this paper is the design of a cache aided communication scheme that serves each user at its own rate. This is done by using spatial multiplexing (instead of multicasting) and hence does not require each packet to be sent at the rate of the weakest user. The proposed system is DoF optimal, but gives significant advantage over previous methods were the rate of the users are different (typically at low and medium signal to noise ratio). We also derive a closed form solution and formulate the optimal power, rate  and cache allocation for a special case of the parameters. 

 Our results indicate a significant improvement in the system throughput due to jointly optimizing cache, power and rate allocation. This is in contrast to the result of \cite{naderializadeh2017optimality}, where it is shown that a separate design of the caching and delivery is order-wise optimal. However, it is worth noting that while we have \emph{total cache size constraint}, the setting in \cite{naderializadeh2017optimality} associated a fixed and uniform cache size to each user, and hence its result does not directly apply to our setting.

The remainder of the paper is organized as follows. In Section~\ref{sec:sysmodel} we present the system model. In Section~\ref{sec: Caching optimization} we formulate that caching optimization problem as a linear program (LP). A closed form solution for the problem for some special range of parameters is presented in Section~\ref{sec: N+M}, followed by 
some numerical results that illustrate the gain offered by caching and our proposed resource allocation method in Section~\ref{sec:NumericalResults}. Finally, we finish the paper by some concluding remarks in Section~\ref{sec: conclusion}.

\section{System model}\label{sec:sysmodel}
\subsection{Network and channel model}
We consider a single cell network with one base station (BS) which is serving $U$ users. The BS has $N_\mathrm{T}$ antennas and each user is equipped with  $N_\mathrm{R}$ antennas. 
We assume a strict fairness setup, in which each user requests   exactly one file, and all files are of the same size (i.e., all users require exactly the same amount of data). We further assume a wideband
communication scheme, in which the bandwidth is divided into $B$ small frequency bins. Symbols are transmitted at the rate of $R_\textrm{S}$ symbols per second, where at each symbol time  one symbol is modulated over each frequency bin without inter symbol interference (e.g., OFDM). Thus, the transmission bandwidth is approximately $B\cdot R_\textrm{S}$.

Considering the time duration of a single symbol, the received sample after matched filtering for the $m$-th frequency bin at $i$-th user is described by an $N_\mathrm{R}\times 1$ vector, given by
\begin{IEEEeqnarray}{rCl}
\mathbf{y}_{i,m}=\mathbf{H}_{i,m}\mathbf{x}_m+\mathbf{w}_{i,m},
\end{IEEEeqnarray} 
where $\mathbf{H}_{i,m} \in \mathbb{C}^{N_\mathrm{R}\times N_\mathrm{T}}$ is the channel matrix between the BS and the $i$-th user in frequency bin $m$, which contains the gain from each BS antenna to each antenna of user $i$, $\mathbf{x}_m$ is the transmitted vector at this frequency bin and $\mathbf{w}_{i,m}\sim \mathcal{CN}(0,\sigma^2 \mathbf{I})$ is the additive complex white Gaussian noise.

We assume a very limited movement for the users during the transmission block. Thus, $r_i$, the distance between the $i$-th user to the BS, does not change. On the other hand, due to small movements of the users, and movements of other objects in the area, each link between two antennas experiences fading. Thus, the channel matrix for the $m$-th frequency bin can be written as:
\begin{IEEEeqnarray}{rCl}\label{e:H_eq_rG}
\mathbf{H}_{i,m}=\sqrt{r_i^{-\alpha}}\cdot \mathbf{G}_{i,m}
\end{IEEEeqnarray}
where $\alpha$ is the path-loss exponent and $\mathbf{G}_{i,m}$ is a random matrix that represents the fading. We consider a rich scattering environment, and hence each link experiences an independent Rayleigh fading. In mathematical terms, we assume that each element of  $\mathbf{G}_{i,m}$ is a proper complex normal random variable, with zero mean and unit variance, and that all elements of the matrices  $\mathbf{G}_{i,m}$ for $i=1,\ldots,U$ are statistically independent. Note that we do not assume any specific model for the frequency dependence of the fading. Yet, we will later assume that the bandwidth is large enough, so that the aggregate rate over all frequency bins mimics  the expected rate.

\subsection{Transmission scheme}

The BS simultaneously transmits different messages to different users. Transmission of some of the requested messages can be ignored if the messages are already stored in the users' cache, as will be detailed later. For other messages, the BS needs to make sure that the transmission of an undesired message will not interfere with the reception of the desired user at each active user. In this work we consider a sub-optimal transmission scheme where the BS transmits each message in a way that causes no interference at all to a group of $N-1$ users. This approach is commonly termed Zero Forcing (ZF) precoding or more specifically, block-diagonalization \cite{spencer2004zero}. To allow this scheme, we assume that the number of transmit and receive antennas satisfy $N_\mathrm{T}=N\cdot N_\mathrm{R}$.

While the block-diagonalization scheme adopted in this work is suboptimal, one should note that it has many merits. On one hand, this scheme is known to asymptotically achieve the optimal DoF in a multi-user MIMO scenario \textcolor{black}{\cite{spencer2004zero}}. On the other hand, it requires a low implementation complexity, and hence is quite popular for practical implementations. Specifically to our work, it is convenient as it results in user rates that are independent of the other users rate and channel (as will be shown below).

To apply the block diagonalization constraint, we use a projection matrix, $\mathbf{P}_{\mathcal{Z}_{i,m}}^\perp$, that projects to the null-space of the channel matrices of $N-1$ selected users, and $\mathcal{Z}_{i,m}$ is the set of users that should not be disturbed by the transmission to user $i$, that is, a transmit message to user $i$ over frequency bin $m$ should be zero-forced at users in $\mathcal{Z}_{i,m}$. Thus, the effective channel of user $i$ at frequency bin $m$ is $\mathbf{H}_{i,m} \mathbf{P}_{\mathcal{Z}_{i,m}}^\perp$, and its achievable rate is:
\begin{IEEEeqnarray}{rCl}
\tilde R_{i,m}= R_\textrm{S} \log_2 \left| \mathbf{I}+\frac{P_i }{\sigma^2}p_{i,m}\mathbf{H}_{i,m} \mathbf{P}_{\mathcal{Z}_{i,m}}^\perp\mathbf{H}_{i,m} ^H \right|,
\end{IEEEeqnarray}
where $|\cdot|$ denotes matrix determinant, $P_i$ is the inter-user power allocation for user $i$, and $p_{i,m}$ is the intra-user power allocation for the $m$-th frequency bin, i.e, $p_{i,m}P_i$ is the effective power allocated to user $i$ in frequency bin $m$. We will assume throughout that the intra-user power  allocation is normalized to $1$ ($\E[ p_{i,m}]=1 \ \forall i$), and the inter-user power allocation is subject to a sum-power constraint, $\sum_i P_i\le P$.

\subsection{Performance evaluation}
We assume that each user can decode its desired message without interference from other messages that were simultaneously transmitted by the BS (i.e., each interfering message is either zero forced by the BS or subtracted using the cache available at the receiver). Thus, the achievable rate for user $i$ is given by:

\begin{IEEEeqnarray*}{rCl}
\tilde R_i&=&\sum_m\tilde R_{i,m}
= \sum_m R_\textrm{S} \log_2 \left|\mathbf{I}+\frac{P_i}{\sigma^2}p_{i,m}\mathbf{H}_{i,m} \mathbf{P}_{\mathcal{Z}_{i,m}}^\perp\mathbf{H}_{i,m} ^H \right|.
\end{IEEEeqnarray*}
Assuming that the bandwidth is \emph{large enough}, it will contain enough fading variations so that we can apply the law of large numbers. Let denote by $B$  the number of frequency bins. Thus, for sufficiently large $B$ the user rate will converge to its expectation:
\begin{IEEEeqnarray}{rCl}\label{e: rate with h and p}
\frac{\tilde R_i}{B}\hspace{-2pt}\rightarrow\hspace{-2pt}\frac{R_i}{B}\hspace{-2pt}=\hspace{-2pt} R_\textrm{S}  \E \left[\log_2 \left|\mathbf{I}\!+\!\frac{P_i}{\sigma^2}p_{i,m}\mathbf{H}_{i,m} \mathbf{P}_{\mathcal{Z}_{i,m}}^\perp\mathbf{H}_{i,m} ^H \right|\right]\!.\ \ \ 
\end{IEEEeqnarray}
Substituting \eqref{e:H_eq_rG} into \eqref{e: rate with h and p}, we have:
\begin{IEEEeqnarray}{rCl}\label{e: final rate}
R_i=B  R_\textrm{S}  \E \left[\log_2 \left|\mathbf{I}+\frac{P_i r_i^{-\alpha}}{\sigma^2}p_{i,m}\mathbf{G}_{i,m} \mathbf{P}_{\mathcal{Z}_{i,m}}^\perp\mathbf{G}_{i,m} ^H \right|\right].
\IEEEeqnarraynumspace
\end{IEEEeqnarray}
For a fixed user $i$ with given $P_i$ and $r_i$, the random quantities $p_{i,m}$, $\mathbf{G}_{i,m}$ and $\mathbf{P}_{\mathcal{Z}_{i,m}}^\perp$ only depend on the channel fading (i.e., matrices $\{\mathbf{G}_{i,m}\}$). Thus, the expectation in \eqref{e: final rate} (which is taken with respect to the fading) depends only on $P_i$ and $r_i$. Hence, we conclude that the user rate depends only on its distance to the BS, $r_i$, and its allocated power, $P_i$. In this setup, it is convenient to characterize each user solely by its achievable rate, $R_i$. 

As an example, in the single receive antenna case  ($N_\mathrm{R}=1$), the product $\mathbf{G}_{i,m} \mathbf{P}_{\mathcal{Z}_{i,m}}^\perp\mathbf{G}_{i,m} ^H$ is a rank-1 matrix (indeed it is an scalar), and its single eigenvalue (denoted as $\rho_{i,m}$) has a standard exponential distribution. 
If we also assume constant intra-user power allocation ($p_{i,m}=1$), the user rate will be
\begin{IEEEeqnarray}{rCl}
R_i&=&B\cdot  R_\textrm{S}  \cdot\E \left[\log_2 \left(1+\eta_i \rho_{i,m}\right)\right]\nonumber\\
&=&-B R_\textrm{S} \log_2 e \cdot e^{1/\eta_i}\cdot \mbox{Ei}(-1/\eta_i),
\end{IEEEeqnarray}
where $\eta_i=\frac{P _ir_i^{-\alpha}}{\sigma^2}$ is the average SNR and $\mbox{Ei}(\cdot)$ is the exponential integral function, defined as $\mbox{Ei}(x)= -\int_{-x}^\infty \frac{e^{-t}}{t} dt$.

\subsection{Caching}
A cache-aided communication scheme includes two phases, namely, placement phase and delivery phase. There is a database of $F$ files, each of a unit length, available at the BS, and each user is interested in one of the files. Each user has an allocated cache, to pre-store some part of the database. During the placement phase, the users' caches are filled with messages (packets) from the database, while the users' requests are not yet revealed. After the placement phase, upon revealing users' demands, the BS transmits a proper set of packets  in order to serve all the users with their desired files. The placement phase occurs in the off-peak time of the network, in order  to improve the communication in the peak-time. Even though our analysis can be applied on general demand profile, in this work we consider the worst case demand scenario, in which users request  distinct files (and consequently, we assume $F\geq U$).


The BS can decide on the best allocation of cache to users, subject to a total cache constraint of $M\cdot F$ units distributed over all users. This optimization  allows the BS to place larger cache at users with lower rates (poor channel conditions) and hence reduce the total transmission time of the BS.

We assume that each user is equally likely to request any file. Thus, without loss of generality, we can simplify the problem by assuming that each user will store similar parts of all files, and hence, the cache contents of the users is invariant under a permutation of the files. Thus, the caching problem reduces to finding the optimal cache placement and the optimal sequence of transmissions that will deliver the desired files to all users in minimal time. 

Note that the cache allocation problem does not depend on many of the systems parameters described above. In fact, it turns out that if the number of files is not less than the number of users  ($F\geq U$), the optimum solution for the caching problem only depends on the number of users $U$, the spatial multiplexing dimension $N$ (the number of users that can be simultaneously served with no interference using only the selected MIMO scheme), the number of copies of the database that are distributed across users' cache $M$, and the communication rates supported by the channel $\{R_i\}_i$.

\section{Caching optimization}\label{sec: Caching optimization}
A  cache aided communication scheme needs to specify which part of each file to be pre-fetched at each user (during the placement phase), and afterwards, given the user requests, what is the transmission scheme that can satisfy the requests of all users (during the delivery phase), i.e., what parts of what files should be jointly transmitted at each stage so that all users will be able to decode all their desired packets. In Subsection~\ref{sub: linear program} we show that this problem can be formulated as a linear optimization problem, and hence can be solved efficiently using linear programming methods. Before that, we give a simple example that illustrates the operation of a valid transmission scheme. 
\subsection{A Simple Example}\label{sub: Example}
Consider a MISO broadcast channel with $N_\mathrm{T}=2$ transmit antennas and  $U=3$ users with single antenna ($N_\mathrm{R}=1$), as shown in Fig.~\ref{fig:MISO}.  We assume a total cache constraint, so that only one copy of each (packet of each) file can be pre-fetched among all the users, i.e., $M=1$. We denote the fraction of files to be cached at user $i$ by $q_i$, which implies $q_1+q_2+q_3=M=1$. Recall that the cache placement is invariant under file relabeling, and hence, $q_i$ faction of  each file in the dataset should be pre-fetched in user $i$. We denote the link capacity of user $i$ by $R_i$. Assume Users 1 and 2 have good channels to support $R_1=R_2=2$ and the third user is further away from the transmitter and can only decode at rate of $R_3=1$.  

\begin{figure}[t!]
    \centering
    \subfloat[File partitioning]{\includegraphics[width = 1.8in]{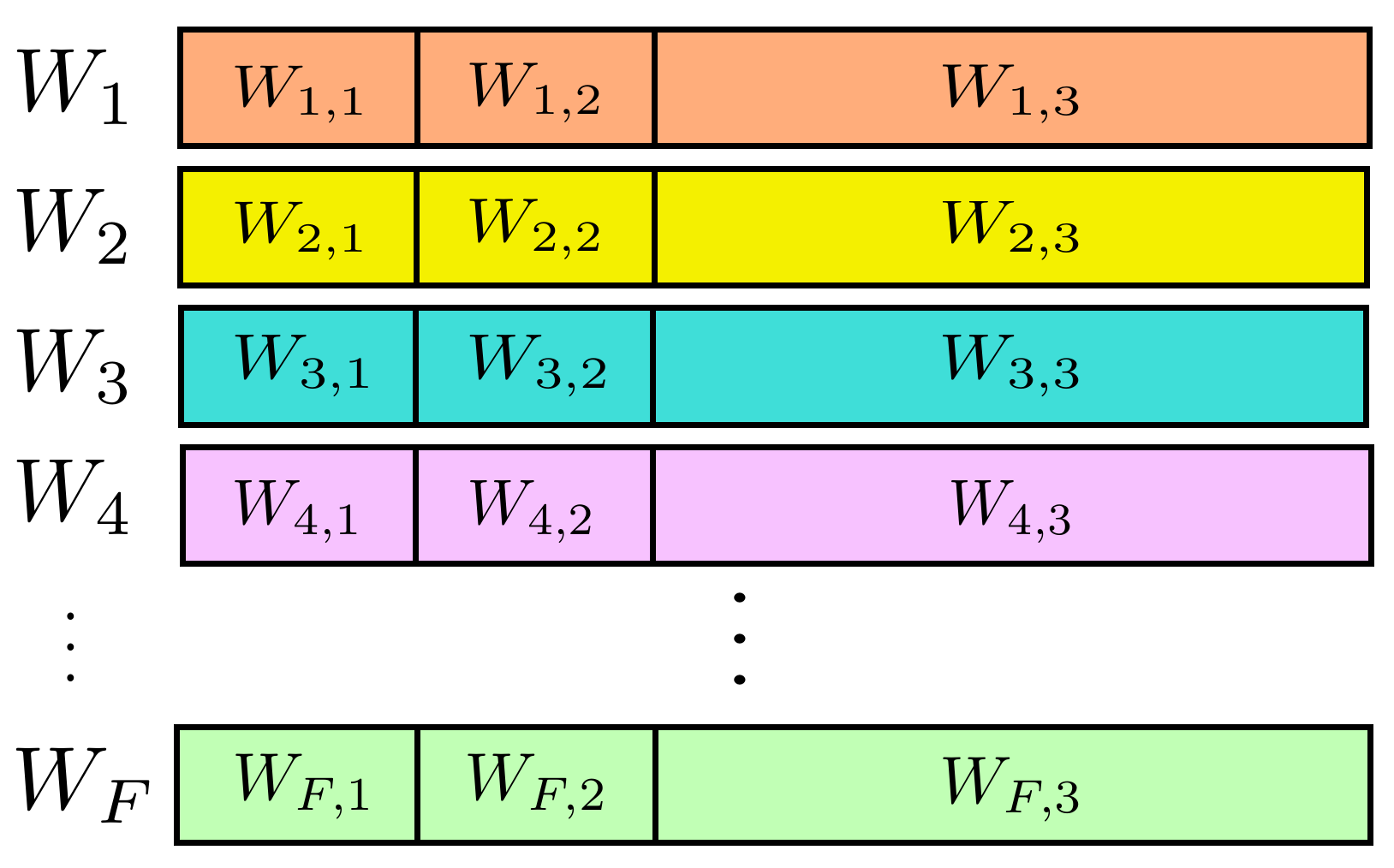}}
    \qquad\qquad
    \subfloat[Cache placement]{\includegraphics[width = 3in
    ]{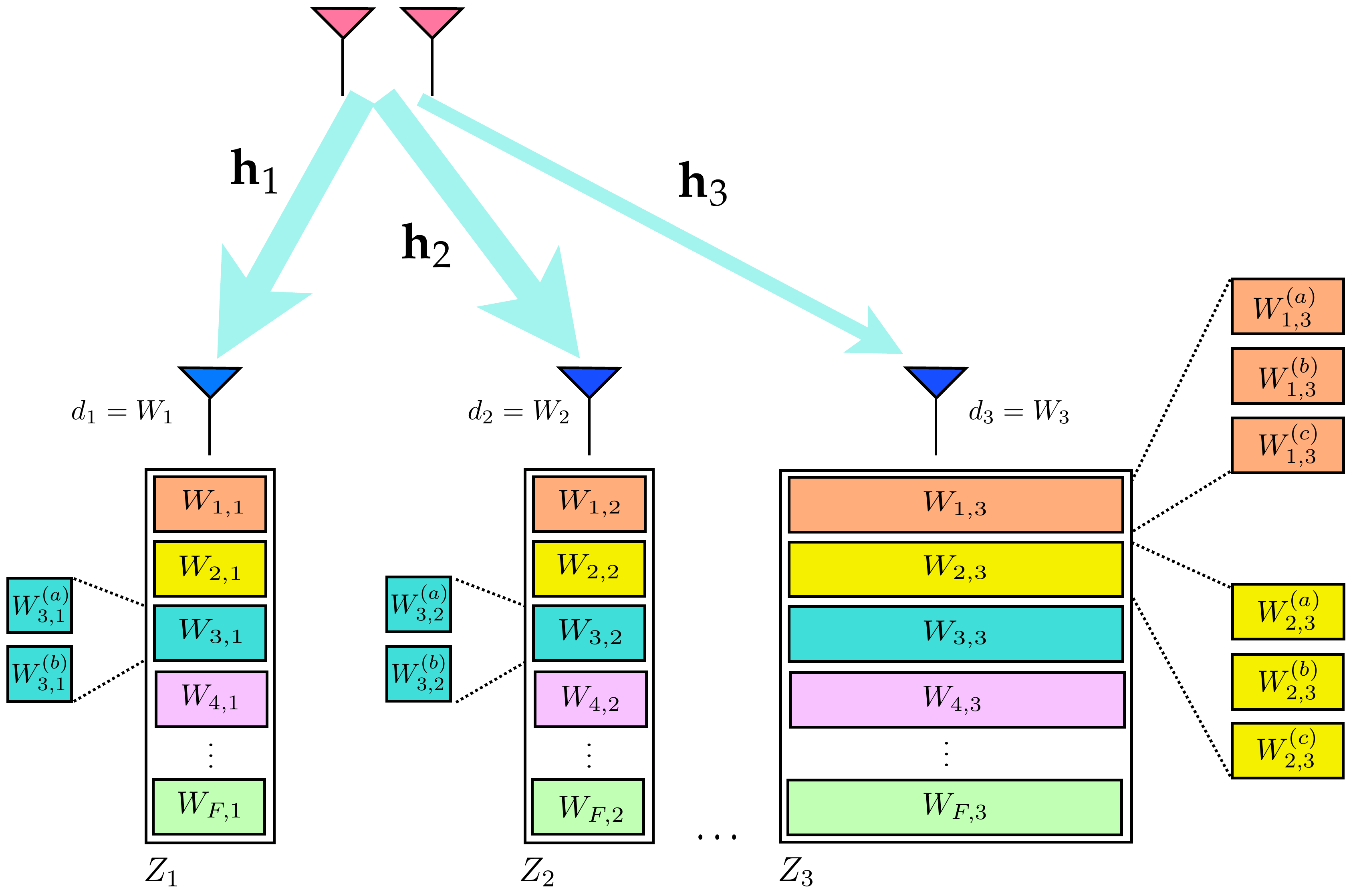}}
    \caption{Cache optimization based on user channel quality}\label{fig:MISO}
\end{figure}

It turns out that the optimum cache allocation to compensate for the weakness of User~3 is $q_1=q_2=1/5$ and $q_3=3/5$. The cache allocation is done by partitioning each file into $3$ \emph{sections}, namely, $W_{k,1}$, $W_{k,2}$, and $W_{k,3}$, which are stored at the cache of Users $1$, $2$, and $3$ respectively, as shown in Fig.~\ref{fig:MISO}(b). The length of cached sub-files will be $\left| W_{k,1} \right|= \left| W_{k,2} \right|=1/5$ and $\left| W_{k,3} \right|=3/5$ (recall that the file lengths are normalized, and hence $1/5$ refers of $1/5$ of a the actual length of the files). 
Note that the cache placement is performed prior to the users request, and hence is identical for all files. In this example we assume that User~1 requested $d_1=W_1$,  User~2 requested $d_2=W_2$ and User~3 requested $d_3=W_3$. 

The delivery phase includes broadcasting $4$ messages. To formally present the broadcast messages, we need to uniformly divide some of the cached sections into smaller \emph{segments} as $W_{3,1}=\left(W_{3,1}^{(a)},W_{3,1}^{(b)}\right)$,  $W_{3,2}=\left(W_{3,2}^{(a)},W_{3,2}^{(b)}\right)$, $W_{1,3}=\left(W_{1,3}^{(a)},W_{1,3}^{(b)},W_{1,3}^{(c)}\right)$, and $W_{2,3}=\left(W_{2,3}^{(a)},W_{2,3}^{(b)},W_{2,3}^{(c)}\right)$, to keep up with the capacity of the links to the users.  Then, each section $W_{k,\cdot}$ and segment  $W_{k,\cdot}^{(\cdot)}$ will be coded to a sequence $\mathbf{s}_{k,\cdot}$ and $\mathbf{s}_{k,\cdot}^{(\cdot)}$, respectively, using a channel code of rate $R_k$. Hence, the length of the resulting sequences will be 
$\left| \mathbf{s}_{1,2}\right|= \left|  \mathbf{s}_{2,1} \right|=\left|\mathbf{s}_{k,\cdot}^{(\cdot)}\right| = \frac{1}{R_k}\left| W_{k,\cdot}^{(\cdot)} \right|  = \frac{1}{10}$.

The sequences needed by the users for successfully decoding their requested files are
\begin{align*}
\centering
\mathrm{User~1: }& \mathbf{s}_{1,2}, \mathbf{s}_{1,3}^{(a)},  \mathbf{s}_{1,3}^{(b)},  \mathbf{s}_{1,3}^{(c)},\\
\mathrm{User~2: }& \mathbf{s}_{2,1}, \mathbf{s}_{2,3}^{(a)},  \mathbf{s}_{2,3}^{(b)},  \mathbf{s}_{2,3}^{(c)},\\
\mathrm{User~3: }& \mathbf{s}_{3,1}^{(a)}, \mathbf{s}_{3,1}^{(b)},  \mathbf{s}_{3,2}^{(a)},  \mathbf{s}_{3,2}^{(b)}.
\end{align*}

All three users can be served  by transmitting:
\begin{align}
\begin{split}
\mathbf{x}(1) &=  \mathbf{h}_2^\perp \mathbf{s}_{1,3}^{(a)} + \mathbf{h}_1^\perp \mathbf{s}_{2,3}^{(a)} +  \mathbf{h}_2^\perp \mathbf{s}_{3,1}^{(a)}, \\
\mathbf{x}(2) &=  \mathbf{h}_3^\perp \mathbf{s}_{1,2} + \mathbf{h}_1^\perp \mathbf{s}_{2,3}^{(b)} +  \mathbf{h}_2^\perp \mathbf{s}_{3,1}^{(b)}, \\
\mathbf{x}(3) &=   \mathbf{h}_2^\perp \mathbf{s}_{1,3}^{(b)} + \mathbf{h}_3^\perp \mathbf{s}_{2,1} + \mathbf{h}_1^\perp \mathbf{s}_{3,2}^{(a)},\\
\mathbf{x}(4) &=  \mathbf{h}_2^\perp \mathbf{s}_{1,3}^{(c)}
+ \mathbf{h}_1^\perp \mathbf{s}_{2,3}^{(c)}  
+\mathbf{h}_1^\perp \mathbf{s}_{3,2}^{(b)},
\end{split}
\end{align}
in four time slots, where $\mathbf{x}(t)$ is beam-forming transmission vectors for the $t$-th time block, which takes $1/10$ time slots. The notation $\mathbf{h}_i^\perp \mathbf{s}_{k,j}$ indicates that all symbols of the codeword $\mathbf{s}_{k,j}$ are precoded over all frequency bins and several symbols, and each pre-coding vector at the $m$-th frequency bin is perpendicular to the $i$-th user channel, $\mathbf{h}_{i,m}$.

Let us consider file retrieval at User~1. For instance, in time  block $t=1$ and frequency bin $m$, User~1 receives $y_{1,m}(1) = \mathbf{h}_{1,m}  \mathbf{x}_m(1) +w_m(1) = \mathbf{h}_{1,m} \mathbf{h}_{2,m}^\perp (\mathbf{s}_{1,3,m}^{(a)}+ \mathbf{s}_{3,1,m}^{(a)}) +w_m(1)$. It removes $\mathbf{s}_{3,1,m}^{(a)}$ using its cache, and then uses the remaining signal to decode $W_{1,3}^{(a)}$. Similarly, each  user can decode all the missing sections of its requested file. 

Note that each transmission takes $1/10$  time slots, and hence the total transmission time is $4/10$, after which, all users have their requested files.  
A total of $3$ files (each of unit length) are delivered to the users, where the network delivered  a total of $4/5+4/5+2/5=2$ files and the remaining sections were already stored at the cache of requesting users. Thus, the throughput of the network is $2/0.4=5$. In contrast, in a similar setting with only single-antenna transmitter, the rate of each packet intended for a subset of users including User~3 should not exceed $R_3=1$. This shows that an optimized coded caching in MISO \emph{offers more gain than just trading antennas vs. cache memory\footnote{Note that many works (e.g., \cite{shariatpanahi2017multi}) evaluate the rate based on the total delivered files (including the parts already cached at the users during the placement phase). In such terminology, the throughput of this network is $3/0.4=7.5$, as $3$ files are delivered in $4/10$ time slots. We use the net throughput in our work in order to emphasize the relation to the physical rates, i.e., the network throughput is $R_1+R_2+R_3=5$.}}.

This example is further illustrated in Subsection III-B, using the terminology of an optimization problem (see Equation (11) and the preceding paragraph).
\hfill  $\square$

\subsection{Cache-aided communication as an optimization problem}\label{sub: linear program}
We next derive a mathematical framework that can describe a cache-aided communication scheme, and show that it can be formulated as a linear programming problem. We focus only on \emph{efficient} transmission schemes, where we define an 'efficient' transmission as one that exploits all degrees of freedom  of the channel. In the setup at hand, an 'efficient' communication must serve $M+N$ users simultaneously at all times\footnote{ Using the notation of \cite{shariatpanahi2017multi}, where there are $\tilde{K}$ users, each with cache size of $\tilde{M}$, and a library of $\tilde{N}$ file, the union of the cache across users holds $\tilde{K}\tilde{M}/\tilde{N}$ copies of the entire data base (similar to our $M$). Using $\tilde{L}$ to denote the number of antennas in \cite{shariatpanahi2017multi}, and the total throughput definition, they showed that the per user DoF is $\frac{\tilde{L} + \tilde{K}\tilde{M}/\tilde{N}}{\tilde{K}(1 - \tilde{M}/\tilde{N})}$. Replacing the notation, and also multiplying by $\tilde K$ for sum-DoF and multiplying  by $(1 - \tilde{M}/\tilde{N})$ to change from total throughput to net throughput (see footnote 1) we get a maximal sum-DoF of $N+M$. Thus, according to \cite{shariatpanahi2017multi}, any 'efficient' scheme will serve $M+N$ users at any time of transmission, and is hence DoF optimal.}. This is done by zero forcing each transmission to $N-1$ direction, and allowing $M$ users to subtract the interference using their cache.

Thus, the content of each transmit message in the network must be stored by $M$ users. In other words, each transmission is intended for a combination of $M+1$ users. Furthermore, as these users store the same parts of all files to allow `efficient' transmissions, this specific user combination must use their cache for the transmission of specific file parts of every other user in the same group. To formulate that, we divide each file into $L$ segments, where each segment is stored by $M$ users. As we have a total of $U$ users and segments are stored by $M$ users, the maximal number of needed sections can be bounded by $L\leq {U \choose M}$. We enumerate these sections by $\ell =1, 2, \dots, L$, and describe them by the row vectors $b_\ell$ for $\ell=1,2,\dots, L$, where $ b_\ell(i)=1$ if user $i$ stores the $\ell$-th  section of each file, and $ b_\ell(i)=0$ otherwise.

The length of the $\ell$-th section is denoted by $u_\ell$. Thus, $0\le u_\ell\le 1$ and $\sum_\ell u_\ell=1$. Note that, the total fraction of each file cached at user $i$ is given by $q_i = \sum_{\ell}  b_\ell(i) u_\ell$, which implies
\begin{align*}
\sum_{i=1}^U q_i 
&= \sum_{i=1}^U \sum_{\ell=1}^L  b_\ell(i) u_\ell 
= \sum_{\ell=1}^L \left(u_\ell \sum_{i=1}^U  b_\ell(i)\right) \nonumber\\
&=  \sum_{\ell=1}^L M u_\ell =M,
\end{align*}
which guarantees that a total of $M$ copies of the entire dataset is distributed among all users. Note that the vectors $b_\ell$ describe the different possibilities for file partitioning, and hence are known in advance (and depend only on $M$ and $U$). The actual allocation is determined by the set of variables $u_\ell$'s which needs to be solved according to the available user rates.

The transmission in each time slot involves a combination of $M+N$ out of the $U$ users, which can be simultaneously served. We will use an index $c$ to label possible combinations where $c=1,\dots, C$ and $C={U \choose M+N}$. Furthermore, each user combination can be active in several transmissions, each with different segments of the file transmitted to each user. The different transmissions for the same user combination ($c$) will be indexed by $j$. 

Each of the $M+N$ users that are active in this time slot receives part of their requested file. We will use  matrices $\mathbf{E}^c_j$ to describe the transmission scheme for a slot, where $(\mathbf{E}^c_j)_{\ell,i}=1$ if the $i$-th user receives (part of) the $\ell$-th file section 
at the $j$-th transmission of user combination $c$, and otherwise $(\mathbf{E}^c_j)_{\ell,i}=0$. Thus, $1\le i \le U$, $1\le \ell \le L$ and $0\le c\le C$. We will next discuss the possible values of $\mathbf{E}^c_j$ and hence the maximal number of transmissions for any user combination.

To characterize the matrices $\mathbf{E}^c_j$, we note that each such matrix satisfies the following conditions: 

\begin{enumerate}
\item[(C1)] Each element in the matrix is either zero or one ($ (\mathbf{E}^c_j)_{\ell,i}\in\{0,1\}$).
\item[(C2)] Each user can receive only one segment at a time, and hence, there is at most one $1$ in each column of $\mathbf{E}^c_j$ (i.e., \textcolor{black}{$\sum_\ell (\mathbf{E}^c_j)_{\ell,i}\in\{0,1\}$}).
\item[(C3)] In an `efficient' transmission, at each time slot there are $N+M$ active users and hence, each $\mathbf{E}^c_j$ matrix contains exactly $M+N$ ones ($\sum_\ell \sum_i (\mathbf{E}^c_j)_{\ell,i}=N+M$).
\item[(C4)] As a user does not need a file segment that is already stored in its cache, we must have $(\mathbf{E}^c_j)_{\ell,i}=0$ for any $\ell$ and $i$ such that $b_\ell(i)=1$ (or alternatively stated: $(\mathbf{E}^c_j)_{\ell,i} b_\ell(i)=0$).
\item[(C5)] Only the users that belong to the user combination $c$ will participate in the reception. Thus, all cache storage indicated by the matrix $\mathbf{E}^c_j$ must be of active users. In other words, if user $i$ is not active in the matrix $\mathbf{E}^c_j$ (that is if $\sum_\ell (\mathbf{E}^c_j)_{\ell,i}=0$), then it is also not used for cache storage ($\sum_\ell b_\ell(i) (\mathbf{E}^c_j)_{\ell,i}=0$).
\end{enumerate}

Thus, for a specific combination of $N+M$ out of  $U$ users, each matrix $\mathbf{E}_j^c$ contains $M+N$ ones in  $M+N$ columns associated to the active users. Each one can be selected independently in its column from the allowed locations (where each column represents a user). In order to count the number of allowed locations of a one in a specific column, we note that (C4) requires the vector $b_\ell$ that corresponds to the row with the one must be zero for this user. Thus we need to count the number of vectors $b_\ell$ that has zero for this user. But, (C5) further limits the allowed locations as it requires that all relevant vectors, $b_\ell$, must have their $M$ ones chosen only from the $M+N$ active users. Thus, we need to count the number of vectors that have $M$ ones out of $N+M-1$ users (the users that are active, excluding the considered user, that must be zero). Hence, there are a total of  $\binom{M+N-1}{M}$ choices for the location of one in each column. 

As the choices of the location of one in each column are independent, and there are  $N+M$ columns of active users, in each $\mathbf{E}_j^c$ matrix, we have a total of 
\setcounter{equation}{8}
\begin{IEEEeqnarray}{rCl}
J=\binom{M+N-1}{M}^{(M+N)}
\end{IEEEeqnarray}
possible matrices for each users combination
(i.e., the range of $j$ is given by $1\le j \le J$).
Denoting by  $T^c_j$ the duration of time required to transmit to a user combination $c$ in mode $j$, the total transmission time is given by:
\begin{IEEEeqnarray}{rCl}
T=\sum_{j,c} T^c_j.
\end{IEEEeqnarray}

\noindent{\bf Revisiting the example}: To demonstrate this formulation, consider the example of Section~\ref{sub: Example}. As $M=1$ and $U=3$ there only $L=\binom{U}{M}=3$ file sections, and three vectors that describe their storage at the different users: $b_1=[1 ,0,0]$, $b_2=[0,1,0]$ and $b_3=[0,0,1]$. The cache placement solution tells us the size of the segments are $u_1=u_2=0.2$ and $u_3=0.6$. For the transmission scheme, we note that this case has only $C=\binom{3}{1+2}=1$ user combination, and a total of $J=\binom{1+2-1}{1}^{(1+2)}=8$ transmission schemes. Out of these, the obtained solution uses only $4$ schemes:
\begin{IEEEeqnarray}{rCl}
\begin{split}
\mathbf{E}^1_1&=\left[\begin{array}{ccc}  
0    & 0 &    1 \\
     0 &    0 &    0\\
     1   &  1 &    0
\end{array}
\right]
,\ \mathbf{E}^1_2=\left[\begin{array}{ccc}  
     0   &  0  &   1\\
     1   &  0  &   0\\
     0   &  1 &    0
\end{array}
\right]
,
\\
 \mathbf{E}^1_3&=\left[\begin{array}{ccc}  
     0   &  1 &    0 \\
     0   &  0 &    1 \\
     1   &  0 &    0
\end{array}
\right]
,\ \mathbf{E}^1_4=\left[\begin{array}{ccc}  
     0   &  0 &    0 \\
     0   &  0 &    1 \\
     1   &  1 &    0
\end{array}
\right],
\end{split}
\end{IEEEeqnarray}
and the transmission time of each mode is $T^1_1=T^1_2=T^1_3=T^1_4=0.1$, implying a total transmission time of $T=0.4$.\hfill $\square$

Recall that transmission to User~$i$ is done at rate $R_i$. In order to serve User~$i$ we have to deliver all non-cached sections of the requested file, that is, $\ell$'s with $b_\ell(i)=0$. Noting that the size of the $\ell$ section is $u_\ell$, we have  
\begin{IEEEeqnarray}{rCl}\label{e: E constraint}
\sum_{j,c} T_j^c (\mathbf{E}_j^c)_{\ell,i}=\frac{u_\ell}{R_i}.
\end{IEEEeqnarray}
for each User~$i$ and section~$\ell$ such that $b_\ell(i)=0$. 
Noting that each vector $b_\ell$ has $U-M$ zero elements, \eqref{e: E constraint} yield in a set of $L(U-M)$ constraints.

Thus, the problem can be formulated as a linear programming minimization:
\begin{IEEEeqnarray}{rCl}
\min_{T^c_j,u_\ell} &&  \sum_{j,c} T^c_j
\notag \\
\text{Subject  to} &:& \sum_\ell u_\ell=1
\notag \\
&& \sum_{j,c} T_j^c (\mathbf{E}_j^c)_{\ell,i}=\frac{u_\ell}{R_i} \quad \forall  \ell,i:b_\ell(i)=0
\notag \\
&& T_j^c\ge0, \quad u_\ell\ge 0\quad \forall c,j,\ell
\label{eq:opt-0}
\end{IEEEeqnarray}
The formulation of a linear programming problem allows us to find an optimal scheme using efficient algorithms. The solution of this problem gives both the details of the cache allocation and placement to all users (through the variables $u_\ell$ and equation $q_\ell= \sum_{\ell} b_\ell(i) u_\ell$), and the sequence of transmission that can deliver the desired files to all requesting  users.


Note that the optimization problem above can easily be adjusted for the case of per user cache size constraint, by adding the  constraint: $\sum_{\ell}  b_\ell(i) u_\ell\le u_{\max}$, where $u_{\max}$ is the maximum  cache size. Yet, in this work we do not pursue this approach, and focus only on the global cache size constraint, as described in Section II.

The number of variables ($u_\ell$'s and $T_j^c$'s) in this problem is 
\setcounter{equation}{13}
\begin{IEEEeqnarray}{rCl}
L+C\cdot J=\binom{U}{M}  +\binom{U}{M+N}\cdot{\binom{M+N-1}{M}}^{(M+N)}.
\end{IEEEeqnarray}
This number grows polynomially with the number of users $U$, but exponentially with $M+N$. Thus, the suggested approach is practical for large network as long as the number of DoF ($M+N$) is not large. Further research is necessary to optimize large networks with large number of antennas or large cache.

On the good side, the number of equality constraints of the linear program is only  $L(U-M)+1$. Thus, an optimal solution includes at most $L(U-M)+1$ non-zero variables \cite{boyd2004convex}. This means that the total number of file sections and transmission modes that are needed for the actual implementation is quite small and limited to ${U \choose M}\cdot(U-M)+1$.

As an example, consider a problem with $U=12$ users with (normalized) rates of $R_i=i$ for $i=1,\ldots,12$. Assume that the BS has $N_\mathrm{T}=2$ antennas, and $M=2$ copies of the dataset are distributed among all the users. The resulting linear programming problem  has  $40161$ variables and $661$ constraints. The actual solution divide the files to only $55$ sections, and uses $421$ transmission modes (i.e., the optimal solution resulted in $39674$ transmission modes that were assigned a zero duration of time). 

The total cache allocated ($q_i$) for each user in this scheme, given by $\sum_\ell u_\ell \cdot (b_\ell)_i$, is (sorted from the user with lowest rate to the user with highest rate): $.63$, $.44$, $.29$, $.20$, $.14$, $.10$, $.07$, $.05$, $.03$, $.03$, $.02$, $0$. The total transmission time for the whole transmission is $T=0.51$. 

As a comparison, a standard algorithm (e.g., mimicking \cite{maddah2014fundamental} for multiple antenna case) that does not account for the different rates, will need to adjust each transmission to the active user with lowest rate. Such algorithm will need more than $T=1.22$ to complete all transmissions. This is more than twice slower  then the proposed algorithm. As another comparison,  using this optimal allocation but with $N=1$ BS antenna instead of $2$ requires $T=0.73$ which is only $40\%$ worse than the $N=2$ case, and still much better than the standard method with even $2$ antennas.

\section{The special case of $U=M+N$}
\label{sec: N+M}
While the linear programming approach allows an efficient optimization of the cache aided communication scheme, it is hard to draw insights from it on the properties of the optimal solution. To get some insights, in the next section we analyze the special case of $U=M+N$. 

\subsection{Cache Allocation}
In this special case, all users are active throughout all the transmissions. This  allows for an analytical performance evaluation, as stated in the following theorem.
\newtheorem{theorem}{Theorem}
\newtheorem{proof}{Proof}
\begin{theorem}\label{th: U=M+N}
For the cache-aided communication problem with $U=M+N$, if the rate of each user satisfies
\begin{IEEEeqnarray}{rCl}\label{rate condition}
R_i \le \frac{1}{N}\sum_{u=1}^U R_u,
\end{IEEEeqnarray}
then the minimal time to serve all users is 
\begin{IEEEeqnarray}{rCl}
T=\frac{N}{\sum_u R_u}. 
\end{IEEEeqnarray}
\end{theorem}
\begin{proof}[Proof of Theorem~\ref{th: U=M+N}]
In the case that $U=M+N$ all `efficient' transmissions must include all users. Thus the total transmission time for each user equals  $T$. Recall that $q_u$ fraction of the file requested by User~$u$ is pre-stored in its cache. Hence, the time required to deliver the remaining $1-q_u$ fraction satisfies
 \begin{IEEEeqnarray}{rCl}\label{c:T}
1-q_u= T  R_u .
\end{IEEEeqnarray}
In addition, the total cache allocated across users  satisfies:
\begin{IEEEeqnarray}{rCl}\label{c:q}
\sum\nolimits_{u=1}^U q_u = M.
\end{IEEEeqnarray}
If the optimization problem in \eqref{eq:opt-0} has a feasible solution, it must  satisfy \eqref{c:T} and \eqref{c:q}. Substituting \eqref{c:T} in \eqref{c:q} gives:
\begin{align*}
\sum_{u=1}^U (1-q_u) = (N+M) - \sum_{u=1}^U q_u= T \sum_{u=1}^U R_u
\end{align*}
which implies 
\begin{align*}
T\geq \frac{N+M - \sum_{u=1}^U q_u}{\sum_{u=1}^U R_u} =\frac{N}{\sum_{u} R_u} .
\end{align*}
However, this solution can be feasible only if the resulting $q_i$'s are feasible. That is,
\begin{align*}
q_u = 1- TR_u \geq 0 
\end{align*}
which is  equivalent to $R_u \leq \frac{1}{T} = \frac{1}{N} \sum_{u=1}^U R_u$. 



The achievability of this result stems from the observation that this scheme is significantly simpler than other caching schemes in the sense that the transmission to each user can be optimized separately. The cache placement in this case only needs to satisfy two simple requirements: 1) Exactly $q_u$ of each file in the database should be stored at user $u$, and 2) The cache content at each user has no overlaps. The optimum transmission scheme always sends to \emph{all} the users simultaneously, and interference management is performed over each individual stream: since each requested packet exists at exactly $M$ users' cache, these users can suppress  the interference using their cache content. Thus, each packet just need to be zero-forced at the $N-1$ users that do not store this packet in their cache.
\end{proof}

\subsection{Power Allocation}
The result of Theorem~\ref{th: U=M+N} implies that a network with $U=N+M$ users can achieve a total throughput of\footnote{Taking into account the files stored in cache, the throughput is defined as $(U-M)/T$ where $T$ is the time needed to complete transmission for all users.} $\sum_{u=1}^{M+N}R_u$. This suggests that the well known water-filling algorithm will be appropriate for throughput optimization. However, the results above also include a condition that the maximal rate should not exceed $\frac{1}{N}\sum_{u=1}^{M+N}R_u$. An intuitive justification for this constraint is the following: A user with a rate that is higher than $\frac{1}{N}\sum_{u=1}^{M+N}R_u$  will be completely served before the other users. Hence, for the remaining transmission time, there are less than $N+M$ active users in the system, and we cannot fully exploit the available DoF of the network.  

This result also represents a fascinating balancing mechanism that brings a natural balance between throughput and fairness. The issue of throughput maximization vs. user fairness has accompanied the field of wireless communication for decades. In most cases, enforcing  fairness reduces the total throughput, and the typical working point is selected as a trade-off between the two. 

In this work, the problem is stated with a strict fairness constraint: each user must receive the same amount of data ($1$ file). Yet, due to the caching, the optimal performance is achieved by maximizing the total throughput. Thus, the caching  allowed a natural balancing of the load between the users. Hence, caching brings two distinct advantages: (1) it increases the network throughput by allowing a simultaneous transmission of $N+M$ users, and, (2) it balances between the data rate of the different users to improve fairness. Note that the second property is obtained  mostly by placing larger cache to poor users, which reduces their communication needs.

The maximum rate constraint represents the cases in which maximal throughput cannot be jointly achieved with the complete fairness. In such cases, the system  needs to allocate more power to poor users in order to further increase their rate and achieve fairness. At the power allocation level, we can allocate the total power among the users to guarantee achievability of the maximum throughput. This can be formally stated as an individual optimization problem stated in \eqref{e:Optimization_problm}.
\begin{align}
\begin{split}
\max_{P_1,\ldots, P_U} &\sum_{u=1}^U R_u
\label{e:Optimization_problm}\\
\textrm{Subject to:  }&R_k\le \frac{1}{N}\sum_{u=1}^U R_u \quad \forall \  k 
\end{split}
\end{align}

The solution for this problem can be obtained by a small adjustment of the standard water-filling algorithm. For simplicity, the algorithm is described only for the case of single antenna per user. Let the power level be denoted by $\rho$. For convenience, we sort the users by their distance from the BS (such that $r_U\le r_{U-1}\ldots\le r_1$) and denote the effective channel gain by $\eta_{i,m}=\mathbf{G}_{i,m} \mathbf{P}_{\mathcal{Z}_{i,m}}^\perp\mathbf{G}_{i,m} ^H$. The optimum water-filling power and rate allocations are given by:
\begin{align}
\begin{split}
P_i p_{i,m}&=\left(\rho-\frac{\sigma^2}{r_i^{-\alpha}\eta_{i,m}}\right)_+,
\\
R^{\textrm{W}}_i&=\E\left[ \log_2 \left(1+\frac{P_i p_{i,m}r_i^{-\alpha}\eta_{i,m}}{\sigma^2}\right)\right],
\end{split}
\end{align}
where $x_+ = \max(x,0)$ is the positive part of $x$. 
If these rates do not satisfy the maximal rate constraint, we need to determine which users will meet the constraint with equality. Noting that these users will always be the users closest to the BS, we just need to determine the number of users for which the constraint will be active. Denoting the number of such users by $h$,  these users will use the constraint rate, $R_{\max}(h)$. Thus, the maximal allowed rate will be
\begin{IEEEeqnarray*}{rCl}
R_{\max}(h)
&=&\frac{1}{N}\sum_{u=1}^{U-h}R^{\textrm{W}}_u+\frac{h}{N}R_{\max}(h)
=\frac{1}{N-h}\sum_{u=1}^{U-h}R^{\textrm{W}}_u
\end{IEEEeqnarray*}
and we need to find  $h$ such that
\begin{IEEEeqnarray}{rCl}
&&R^{\textrm{W}}_i\le R_{\max}(h)\quad \forall \ i<U-h,
\\
&&R^{\textrm{W}}_i\ge R_{\max}(h)\quad \forall \ i\ge U-h.
\end{IEEEeqnarray}
It is easy to show that there will always exist exactly one value of $0\le h\le N-1$ that satisfies both inequalities.

After determining $h$, we can find the power required by each user, and hence the sum power of the BS. Iterating over the initial power level will give the appropriate power level that matches the total power $P$ (noting that the total power is monotonically increasing with the level of the water $\rho$). 

\section{Numerical Results}
\label{sec:NumericalResults}


In this section we present numerical results to  better illustrate the gain obtained by cache-aided communication scheme and the proposed optimization framework. The performance presented here are based on Monte Carlo simulation with $1000$ network realizations per point. Each Network realization consists of a random positioning of the $U=M+N$ users independently and uniformly over a circular area of radius $r_{\max}$. The channel gain were evaluate using the distances to the BS, and a random generation of $B=100$ Rayleigh fading variables per user. In all simulations we considered single antenna users ($N_\mathrm{R}=4$) and a BS with $N_\mathrm{T}=N=4$ antennas.

\begin{figure}
  \begin{center}
\includegraphics[width=0.45\textwidth]{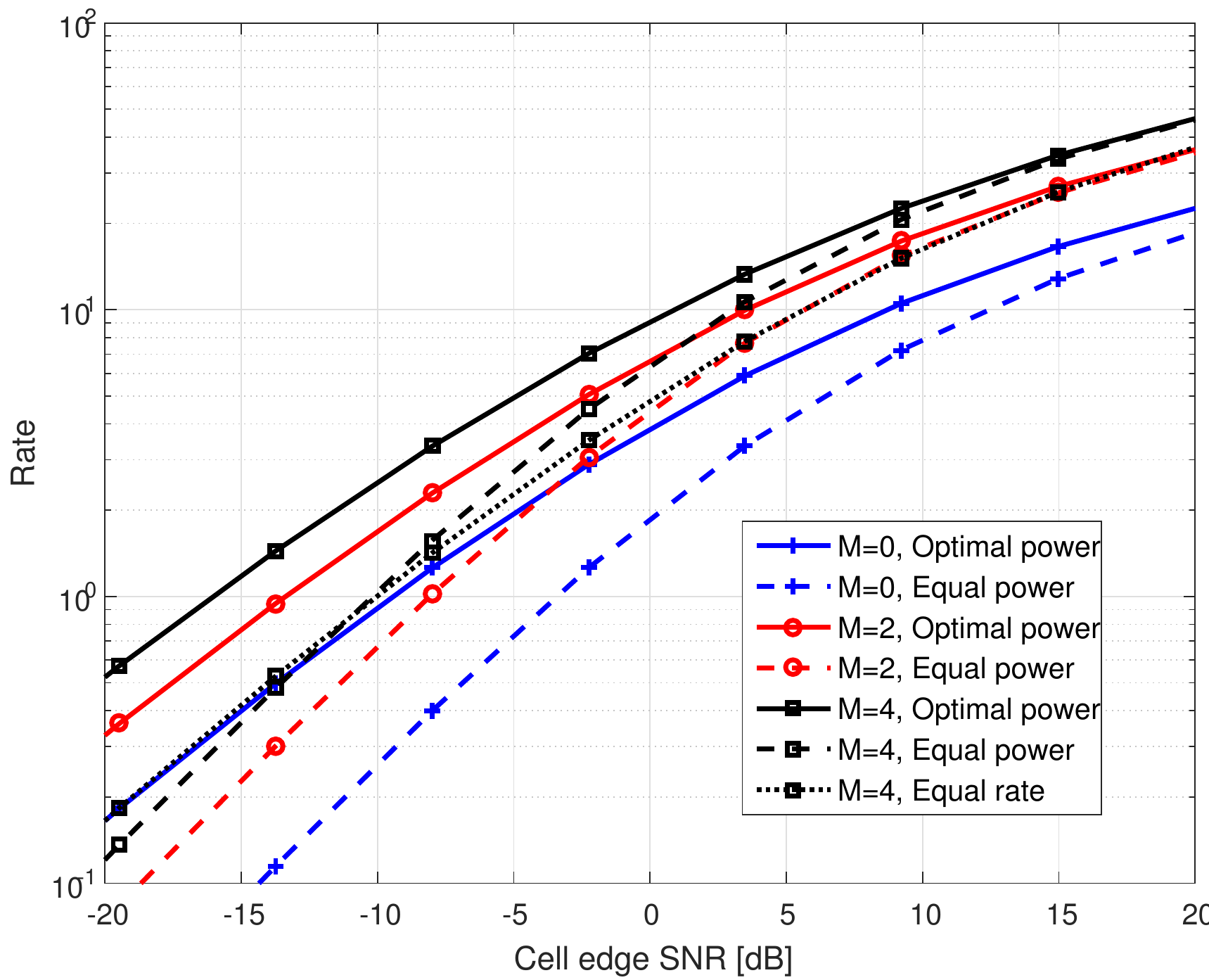}
  \end{center}
\caption{Total network throughput as a function of the received power at cell edge for various cache sizes.}\label{fig:rate_vs_M}
\end{figure}


Fig.~2 depicts the total network throughput as a function of  SNR at the cell edge\footnote{Recall that each user experiences a different SNR. Thus, the SNR at the cell edge is a convenient reference point, even though no user is actually located at the cell edge.}.
In this study, the throughput is defined by the data that is delivered to the users during the delivery phase (not including the data that was previously placed in their cache). The figure depicts the performance for the cases that the overall cache memory at all users contains $M=2$ or $M=4$ copies of the entire database. For reference, the figure also depicts the performance with no cache ($M=0$). Following our assumption in Section~IV, the number of users changes according to the allocated cache size so that $U=M+N$.
The figure depicts the performance with and without caching for three  types of resource allocation. `Optimal power' depicts the performance with the optimal power allocation and optimal cache allocation as described in Section IV. `Equal power' uses the same power for all frequency bins of all users, but optimal cache allocation. `Equal rate'  uses the power allocation that provide equal rate all users (with optimal water-filling intra-user power allocation for the different frequency bins of all users).

Note that for the `Equal rate'  the optimal cache allocation is equal for all users. This scheme characterizes the performance of previously published schemes in which the transmission rate is taken as the minimal achievable rate among the active users (e.g.,  \cite{shariatpanahi2017multi, huang2015performance, ibrahim2017optimization}). Obviously, for such schemes were the performance is bounded by the minimal rate, the optimal power allocation leads to `Equal rate'.

The figure shows that caching and optimal power allocation improve the performance. Yet, the fine details are hard to observe due to the large range of the vertical axis. Fig.~\ref{fig:norm_rate_vs_M} presents the same rates, but in a normalized manner that allows a better inspection. In this figure, each of the total rates  was divided by the total rate in the case of optimal power allocation with no cache available at the users.

\begin{figure}
  \begin{center}
\includegraphics[width=0.45\textwidth]{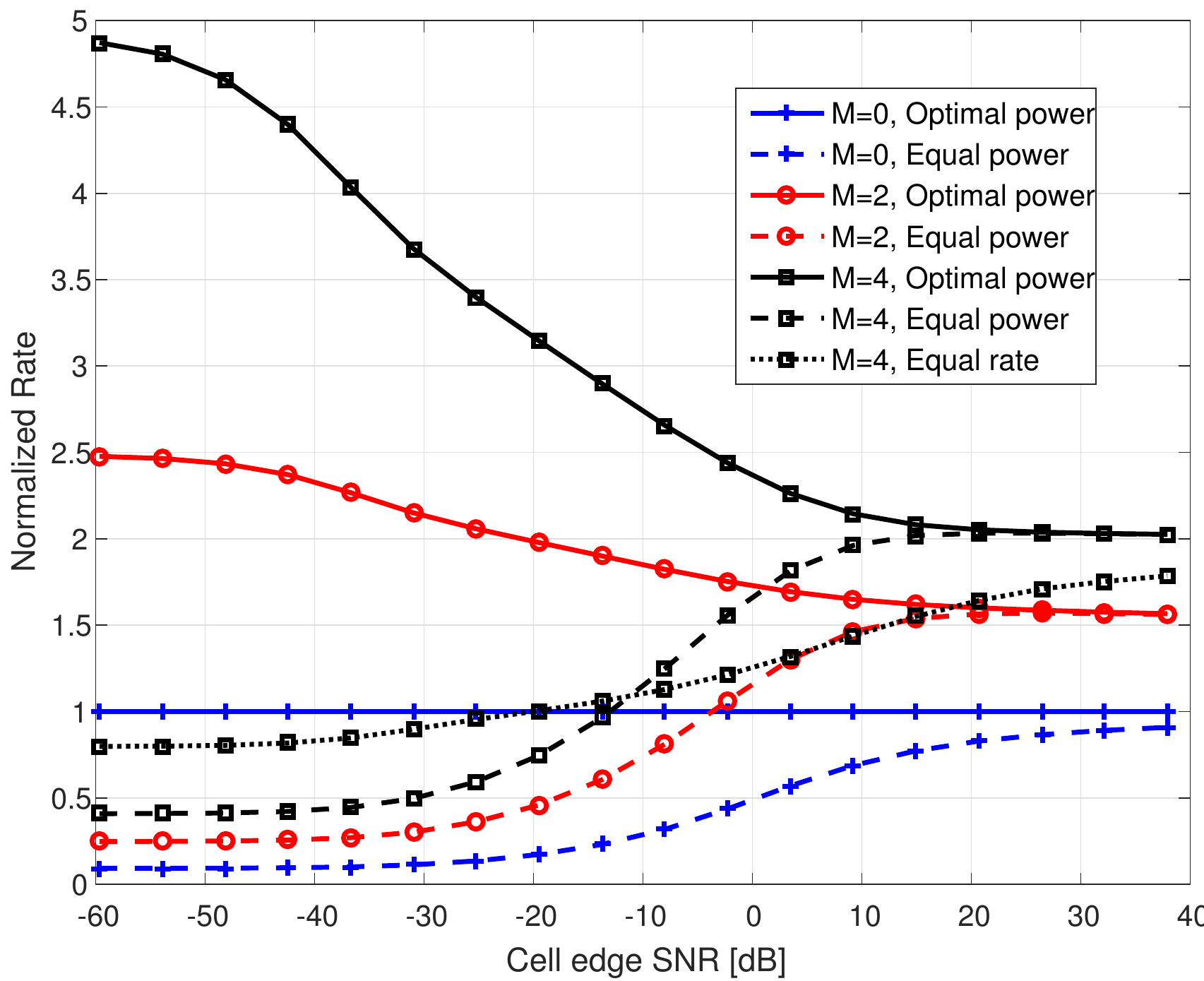}
  \end{center}
\caption{Normalized network throughput as a function of the SNR at cell edge for various cache sizes.}\label{fig:norm_rate_vs_M}
\end{figure}

At high SNR regime, the difference between the channel gain of the different users becomes negligible, and all users approach the same rate. Hence, rate balancing is not critical and we only see the effect of throughput increase. As each scheme allows for serving  $M+N$ users simultaneously, we expect a gain of $(M+N)/N$, which is $1.5$ and $2$ for the case of $M=2$ and $M=4$, respectively. We see that these values are indeed achieved with or without optimal power allocation. This shows that optimal power allocation has a minimal effect on the overall throughput in this regime. 

On the other hand, at low SNR regime, the difference between user rates is significant, and we also see the effect of the rate balancing. It is transparent that the inherent rate balancing effect of the optimized caching scheme leads to a significant increase in the network throughput. Thus, the ability to maximize the throughput while keeping strict fairness gives gains which are  close to twice the throughput gains.

\begin{figure}
  \begin{center}
\includegraphics[width=0.45\textwidth]{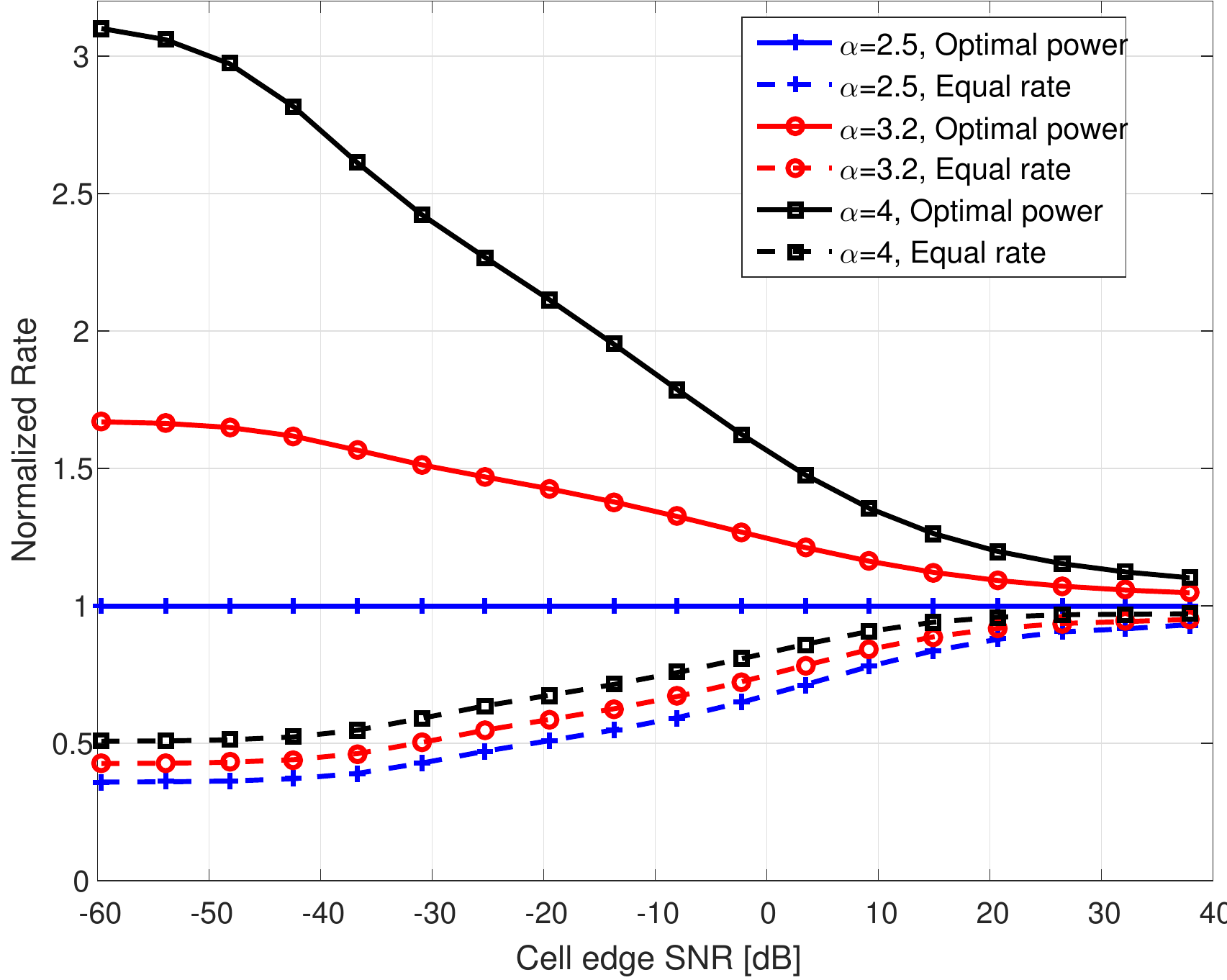}
  \end{center}
\caption{Normalized network throughput as a function of the SNR at cell edge for various path loss exponents.}\label{fig:norm_rate_vs_alpha}
\end{figure}

Note that for $M=0$, the maximum rate constraint in  \eqref{e:Optimization_problm} 
requires that all rates to be equal. This is reasonable as in the absence of cache, we have no balancing  mechanism. Thus, in this case the powers must be set such that all users achieve exactly the same rate. Hence, the performance of the `Optimal power' and the `Equal rate' schemes for $m=0$ are identical. In comparison, the `Equal power' scheme achieves lower rates due to the intra-user power allocation, i.e., the less efficient use of the frequency bins with good fading. Also note that at very low SNR, the `Equal rate' scheme suffers a small decrease in performance with larger cache sizes. This is due to the strict fairness constraint, that requires the system to bring exactly the same rate to a larger number of users simultaneously, while the increase in DoF is meaningless at such low SNRs.


Fig.~\ref{fig:norm_rate_vs_alpha} shows the normalized rate for various values of the path loss exponent, $\alpha$. Again, at high SNR, the rates are almost identical, and all schemes have similar performance. However, for low SNR we see significant difference between the curves. We note that the amplitude variations between the users are more considerable for larger values of the path loss exponent. Thus, a more significant gain of the balancing mechanism can be observed for larger values of $\alpha$. In particular, we see the largest gain for $\alpha=4$, and the second largest gain for $\alpha=3.2$.

\section{Conclusion}
\label{sec: conclusion}

In this paper we studied the cache-aided communication problem for cellular networks. An important feature of the considered model is availability of multiple antennas at the base station and the users. The links between the BS and users are assumed to be asymmetric, and are modeled by wideband fading channels. While it is known that cache and spacial diversity can be traded to achieve DoF, our analysis is not limited to DoF, and we have studied the time of delivery in finite signal-to-noise ratio regime. 

We formulated the cache allocation, cache placement, and delivery scheme as a joint linear program. Even though the number of variables in the LP is large (exponential in problem parameter), the solution is very sparse (the number of  non-zero variables is quadratic in problem parameters), which makes it feasible for practical implementation. 
The suggested scheme is better than previously known schemes as each user can be served at the rate of its own channel, rather than being compromised by the poorest user in the communication group. 

We also considered a special case of the parameters for which  a closed form solution can be obtained. This closed form solution was used to derive the optimal power allocation algorithm. It is shown that the joint optimization of cache usage and power allocation yields a gain for MIMO coded-caching, which goes  beyond the throughput increase. In particular, it is shown that in this case, the cache is used to balance the users such that fairness and throughput are no longer contradicting. More specifically, in this case, strict fairness is achieved jointly with maximizing the network throughput. 

\bibliographystyle{IEEEtran}
\bibliography{Refs}

\end{document}